\documentclass[letterpaper, 11pt]{article}
\usepackage{appendix}
 
\newcommand{\comment}[1]{}
 
\usepackage{graphicx}
\usepackage{subfigure}
\usepackage{amsmath}
\usepackage{mdwlist}
\usepackage{pxfonts}
 
\usepackage{algorithm}

\usepackage{hyperref}

\usepackage{xspace}

\newcommand{\sv}{{\mathcal V}}
\newcommand{\se}{{\mathcal E}}

\newenvironment{proof}{\paragraph{\bf Proof:}}{\hspace*{\fill}\(\Box\)}

\newtheorem{theorem}{Theorem}

\newtheorem{claim}{Claim}

\newtheorem{definition}{Definition}

\newtheorem{lemma}{Lemma}

\def\noflash#1{\setbox0=\hbox{#1}\hbox to 1\wd0{\hfill}}

\newcommand{\scriptf}{\mathcal{F}}
\newcommand{\scripte}{\mathcal{E}}
\newcommand{\scriptv}{\mathcal{V}}

\begin{document}
\title{Parameter-independent Iterative Approximate Byzantine Consensus \footnote{\normalsize This research is supported in part by National Science Foundation award CNS 1059540
and Army Research Office grant W-911-NF-0710287. Any opinions, findings,
and conclusions or recommendations expressed here are those of the authors and do not
necessarily reflect the views of the funding agencies or the U.S. government.}}

\author{Lewis Tseng$^{1,3}$ ~~and~~ Nitin Vaidya$^{2,3}$\\~\\
 \normalsize $^1$ Department of Computer Science,\\
 \normalsize $^2$ Department of Electrical and Computer Engineering, 
 and\\ \normalsize $^3$ Coordinated Science Laboratory\\ \normalsize University of Illinois at Urbana-Champaign\\ \normalsize Email: \{ltseng3, nhv\}@illinois.edu~\\~\\Technical Report}

\date{August 23rd, 2012}
\maketitle

 
 
 
\setcounter{page}{1}

\begin{abstract}
In this work, we explore iterative approximate Byzantine consensus algorithms that do not make explicit use of the global parameter of the graph, i.e., the upper-bound on the number of faults, $f$.

\end{abstract}

\newpage

\section{Introduction}
\label{sec:intro}
We consider ``iterative'' algorithms for achieving approximate Byzantine consensus in synchronous point-to-point communication networks that are modeled by arbitrary {\em directed}\, graphs. 
The {\em iterative approximate Byzantine consensus} (IABC) algorithms of interest have the following properties:

\begin{itemize}
\item {\em Initial state} of each node is equal to a real-valued {\em input} provided to that node.
\item {\em Validity} condition: After each iteration of an IABC algorithm, the state of each fault-free node must remain in the {\em convex hull} of the states of the fault-free nodes at the end of the {\em previous} iteration.\footnote{See Section \ref{s_discussion} for a variation on the validity condition.}
\item {\em Convergence} condition: For any $\epsilon>0$, after a sufficiently large number of iterations, the states of the fault-free nodes are guaranteed to be within $\epsilon$ of each other.
\end{itemize}
In this paper, we are interested in {\em parameter-independent} algorithms that
do not require explicit knowledge of the upper bound on the number of faults to be tolerated.
In particular,
we introduce a specific parameter-independent IABC algorithm, named {\em Middle} Algorithm. We derive a necessary condition on the underlying communication graph under which the {\em Middle} algorithm can tolerate up to $f$ Byzantine faults. For graphs that satisfy this necessary condition, we show the correctness of {\em Middle} Algorithm, proving that our necessary condition is tight.

For a more thorough discussion on related work, please refer to our previous work \cite{vaidya_PODC12}.

\section{System Model}
\label{s_models}

{\em Communication model:}
The system is assumed to be {\em synchronous}. The communication network is modeled as a simple {\em directed} graph $G(\scriptv,\scripte)$, where $\scriptv=\{1,\dots,n\}$ is the set of $n$ nodes, and $\scripte$ is the set of directed edges between the nodes in $\scriptv$. 
With a slight abuse of terminology, we will use the terms {\em edge} and {\em link} interchangeably.
We assume that $n\geq 2$, since the consensus problem for $n=1$ is trivial. Node $i$ can reliably transmit messages to node $j$ if and only if the directed edge $(i,j)$ is in $\scripte$. Each node can send messages to itself as well, however, for convenience, we {exclude self-loops} from set $\scripte$. That is, $(i,i)\not\in\scripte$ for $i\in\scriptv$. 

For each node $i$, let $N_i^-$ be the set of nodes from which $i$ has incoming
edges. That is, $N_i^- = \{\, j ~|~ (j,i)\in \scripte\, \}$. Similarly, define $N_i^+$ as the set of nodes to which node $i$ has outgoing edges. That is, $N_i^+ = \{\, j ~|~ (i,j)\in \scripte\, \}$. Nodes in $N_i^-$ and $N_i^+$ are, respectively, said to be incoming and outgoing neighbors of node $i$. Since we exclude self-loops from $\scripte$, $i\not\in N_i^-$ and $i\not\in N_i^+$.  However, we note again that each node can indeed send messages to itself. \\

{\em Failure Model:}
We consider the Byzantine failure model, with up to $f$ nodes becoming faulty. A faulty node may {\em misbehave} arbitrarily. Possible misbehavior includes sending incorrect and mismatching (or inconsistent) messages to different neighbors. The faulty nodes may potentially collaborate with each other. Moreover, the faulty nodes are assumed to have a complete knowledge of the execution of the algorithm, including the states of all the nodes, contents of messages the other nodes send to each other, the algorithm specification, and the network topology.

\section{Middle Algorithm}
\label{sec:middle}

The {\em Middle} algorithm is an iterative approximate Byzantine consensus (IABC) algorithm, and its structure is similar to other algorithms studied in prior work \cite{AA_Dolev_1986, AA_nancy, vaidya_PODC12}.
Each node $i$ maintains state $v_i$, with $v_i[t]$ denoting the state of node $i$ at the {\em end}\, of the $t$-th iteration of the algorithm ($t \geq 0$). Initial state of node $i$, $v_i[0]$, is equal to the initial {\em input}\, provided to node $i$. At the {\em start} of the $t$-th iteration ($t>0$), the state of
node $i$ is $v_i[t-1]$. The {\em Middle} algorithm requires each node $i$ to perform the following three steps in iteration $t$, where $t>0$. Note that the faulty nodes may deviate from this specification. 

\vspace*{4pt}
\hrule
\vspace*{2pt}
{\bf Middle Algorithm}
\vspace*{2pt}\hrule

\begin{enumerate}

\item {\em Transmit step:} Transmit current state $v_i[t-1]$ on all outgoing edges.

\item {\em Receive step:} Receive values on all incoming edges. These values form vector $r_i[t]$ of size $|N_i^-|$.
 
When a fault-free node expects to receive a message from a neighbor but does not receive the message, the message value is assumed to be equal to some {\em default value}.

\item {\em Update step:}

\begin{itemize}
\item Sort the values in $r_i[t]$ in an increasing order with ties being broken arbitrarily,
and use the sorted order of values
to form a partition of nodes in $N_i^-$ into sets $B,M,T$ as follows:
(i) set $B$ contains nodes from whom the smallest $\lfloor |N_i^-|/3 \rfloor$
values in the sorted $r_i[t]$ are received, 
(ii) set $T$ contains nodes from whom the largest $\lfloor |N_i^-|/3 \rfloor$
values in the sorted $r_i[t]$ are received, and 
(iii) set $M$ contains the remaining nodes from whom the values in the ``middle'' of sorted $r_i[t]$ are received. That is, $M=N_i^--B-T$.
 \footnote{For sets $X$ and $Y$, $X-Y$ contains elements that are in $X$ but not in $Y$. That is, $X-Y=\{i~|~ i\in X,~i\not\in Y\}$.} 
Thus, $|M|=|N_i^-|-2\lfloor |N_i^-|/3 \rfloor$.

\item 
Let $w_j$ denote the value received from node $j \in M$. For convenience, define $w_i=v_i[t-1]$ to be the value node
$i$ ``receives'' from itself. Observe that if $j\in \{i\}\cup M$ is fault-free, then $w_j=v_j[t-1]$.

\item Define
\begin{eqnarray}
v_i[t] ~ = ~\sum_{j\in \{i\}\cup M} a_i \, w_j
\label{e_Z}
\end{eqnarray}
where
\[ a_i ~=~ \frac{1}{|M|+1}~=~ \frac{1}{|N_i^-|- 2\lfloor |N_i^-|/3 \rfloor+1}
\] 

The ``weight'' of each term on the right-hand side of
(\ref{e_Z}) is $a_i$, and these weights add to 1. Also, $0<a_i\leq 1$.

For future reference, let us define $\alpha$ as:
\begin{eqnarray}
\alpha = \min_{i\in \scriptv}~a_i
\label{e_alpha}
\end{eqnarray}

\end{itemize}
\end{enumerate}
\hrule

We now define $U[t]$ and $\mu[t]$, assuming that $\scriptf$ is the set of Byzantine faulty nodes, with the nodes in $\scriptv-\scriptf$ being  fault-free.

\begin{itemize}
\item $U[t] = \max_{i\in\scriptv-\scriptf}\,v_i[t]$. $U[t]$ is the largest state among the fault-free nodes at the end of the $t$-th iteration. Since the initial state of each node is equal to its input, $U[0]$ is equal to the maximum value of the initial input at the fault-free nodes.

\item $\mu[t] = \min_{i\in\scriptv-\scriptf}\,v_i[t]$. $\mu[t]$ is the smallest state among the fault-free nodes at the end of the $t$-th iteration. $\mu[0]$ is equal to the minimum value of the initial input at the fault-free nodes.
\end{itemize}

The {\em Middle} algorithm is correct if it satisfies the following conditions in the
presence of up to $f$ Byzantine faulty nodes:

\begin{itemize}
\item {\em Validity:} $\forall t>0,
~~\mu[t]\ge \mu[t-1]
~\mbox{~~and~~}~
~U[t]\le U[t-1]$

\item {\em Convergence:} $\lim_{\,t\rightarrow\infty} ~ U[t]-\mu[t] = 0$
\end{itemize}
The objective in this paper is to identify the necessary and sufficient
conditions for Middle algorithm to satisfy the above validity and convergence conditions for a given $G(\scriptv,\scripte)$.

\section{Necessary Condition}
\label{sec:necessity}

For the {\em Middle} algorithm to be correct, the network graph $G(\scriptv,\scripte)$ must satisfy the necessary condition proved in this section. We first define relations $\Rightarrow$ and $\not\Rightarrow$ that are used frequently in our discussion.

\begin{definition}
\label{def:absorb}
For non-empty disjoint sets of nodes $A$ and $B$,

\begin{itemize}
\item $A \Rightarrow B$ iff there exists a node $v\in B$ such that 

\begin{equation}
\label{eq:absorb}
\frac{|N_v^- \cap A|}{|N_v^-|} > \frac{1}{3}
\end{equation}

\item $A\not\Rightarrow B$ iff $A\Rightarrow B$ is {\em not} true. \\
\end{itemize}
\end{definition}

\begin{theorem}
\label{thm:nc}
Suppose that Middle Algorithm is correct in graph $G(\scriptv,\scripte)$ in the presence of up to $f$ Byzantine faults. Then, both the following conditions must be true:

\begin{itemize}
\item For every node $v \in \scriptv$, $|N_v^-| \geq 3f$.
\item Let sets $F,L,C,R$ form a partition\footnote{Sets $X_1,X_2,X_3,...,X_p$ are said to form a partition of set $X$ provided that (i) $\cup_{1\leq i\leq p} X_i = X$, and (ii) $X_i\cap X_j=\Phi$ if $i\neq j$.} of $\scriptv$, such that $L$ and $R$ are both non-empty, and $|F|\leq f$. Then, either $C\cup R\Rightarrow L$, or $L\cup C\Rightarrow R$.
\end{itemize}
\end{theorem}

\begin{proof}

\paragraph{\em Proof of first condition:}
The first condition is trivially true when $f=0$. Thus, let us now assume that $f\geq 1$.
Suppose by way of contradiction that there exists a node $i$ such that $|N_i^-| < 3f$. 
Consider two cases in iteration 1:

\begin{itemize}
\item $|N_i^-| = 0$:
Suppose that node $i$ has initial input of $X$, and all the remaining nodes have input $x$, where $x<X$. Since node $i$ has no incoming edges, clearly, $v_i[1]=X$.

Consider two cases:
\begin{itemize}
\item There exists a node $j\neq i$ such that $(i,j)\in\scripte$, and the in-degree of node $j$ is such that the value $X$ is {\em not} eliminated in the {\em Update} step, i.e., $|N_j^-| \leq 2$: In this case, $v_j[1]>x$ since $X>x$. However, in the event that node $i$ is actually faulty, $v_j[1]$ will not satisfy the validity condition, since the initial inputs at all the fault-free nodes are all $x$ (if node $i$ were to be faulty).

\item For each node $j\neq i$, either $(i,j)\not\in\scripte$, or 
$(i,j)\in\scripte$ but the value received from node $i$ is dropped at node $j$ during the {\em Update} step:
In this case, all the values that affect the new state of node $j$ are $x$, and
$v_j[1]=x$. It is easy to see that the same scenario will repeat in each iteration, violating convergence condition when all the nodes (including $i$) are fault-free ($v_i$ remains at $X$, and for each node $j \neq i, v_j$ remains at $x$).
\end{itemize}

\item $|N_i^-| \geq 1$:
Assume that $min(f,|N_i^-|)$ incoming neighbors of node $i$ are faulty, and that all the remaining nodes are fault-free. Let $F$ denote the set of faulty nodes.
Note that $|F|\geq 1$.

Let $R = \scriptv - \{i\}-F$. Consider the case when (i) each node in $R$ has input $x$, and (ii) node $i$ has input $X > x$. In the {\em Transmit} step of iteration 1, suppose that the faulty nodes in $F$ send a sufficiently large value
$Y$
(elaborated below) on outgoing links to node $i$, and send value $x$ on outgoing links to nodes in $R$. This behavior is possible since nodes in $F$ are faulty. Each fault-free node $k \in \scriptv - F$ sends $v_k[0]$ (its input) on all its outgoing links.

Since $|N_i^-|<3f$, set $M$ at node $i$ in iteration 1 contains at least one value received from a faulty incoming neighbor. Then it is easy to see that
the faulty nodes can choose $Y$ such that $v_i[1] > X$. Since $i$ is fault-free,
and $v_i[1]$ exceeds the initial input at all the fault-free nodes, the
validity condition is violated.
\end{itemize}
In all cases above, either validity or convergence is violated, contradicting the assumption that the {\em Middle} algorithm is correct in the given graph.

\paragraph{\em Proof of second condition:}
Since the first condition is already proved to be necessary, we assume that the graph satisfies that condition.
The proof for the second condition is also by contradiction.
Suppose that the second condition is violated, i.e., in $G$, there exists some partition $F,L,C,R$ such that $|C \cup R| \not\Rightarrow L$ and $|L \cup C| \not\Rightarrow R$. Thus, for any $i \in L$, $\frac{|N_i^- \cap (C \cup R)|}{|N_i^-|} \leq \frac{1}{3}$, and for any $j \in R$, $\frac{|N_j^- \cap (L \cup C)|}{|N_j^-|} \leq \frac{1}{3}$.

Also assume that the nodes in $F$ (if non-empty) are all faulty, and the nodes in $L,R,C$ are all fault-free. Note that the fault-free nodes are not aware of the true identity of the faulty nodes.

Consider the case when (i) each node in $L$ has initial input $x$, (ii) each node in $R$ has initial input $X$, such that $X > x$, and (iii) each node in $C$ (if non-empty) has an input in the interval $(x,X)$.

In the {\em Transmit} step of iteration 1, suppose that each faulty node in $F$ (if non-empty) sends $x^- < x$ on outgoing links to nodes in $L$, sends $X^+ > X$ on outgoing links to nodes in $R$, and sends some arbitrary value in interval $[x,X]$ on outgoing links to nodes in $C$ (if non-empty). This behavior is possible since nodes in $F$ are faulty. Note that $x^-<x<X<X^+$. Each fault-free node $k \in \scriptv-F$ sends $v_k[0]$ to nodes in $N_k^+$ in iteration 1.

Consider a node $i \in L$. In iteration 1, node $i$ receives $x^-$ from the nodes in $N_i^- \cap F$, $x$ from the nodes in $\{i\}\cup (N_i^- \cap L)$, and values in $(x,X]$ from the nodes in $N_i^- \cap (C \cup R)$. Then in the {\em Update} step, $|B| \geq f \geq |F|$ due to the first condition, i.e., $|N_i^-| \geq 3f$. Furthermore, set $T$ (calculated in the {\em Update} step at node $i$) contains all the values from $N_i^- \cap (C \cup R)$, since $|C \cup R| \not\Rightarrow L$, i.e., $\frac{|N_i^- \cap (C \cup R)|}{|N_i^-|} \leq \frac{1}{3}$,
and the values received from the nodes in $C\cup R$ are the largest values
in vector $r_i[1]$.
 Recall that in the {\em Update} step, node $i$ would eliminate sets $B$ and $T$, and the remaining values, i.e., values in $\{i\} \cup M$, are all $x$, and therefore, $v_i[1]$ will be set to $x$ as per (\ref{e_Z}).

Thus, $v_i[1] = x$ for each node $i\in L$. Similarly, we can show that $v_j[1] = X$ for each node $j \in R$. Now consider the nodes in set $C$ (if non-empty).
The initial state of nodes in $C$ is in $(x,X)$, and all the values received
from the neighbors are in $[x,X]$, therefore, their new state of the nodes
in $C$ will remain in $(x,X)$ when using the {\em Middle} algorithm
(since the node's own state is assigned a non-zero weight in (\ref{e_Z})).

The above discussion implies that, at the end of iteration 1, the following conditions hold true: (i) state of each node in $L$ is $x$, (ii) state of each node in $R$ is $X$, and (iii) state of each node in $C$ is in the interval $(x,X)$. These conditions are identical to the initial conditions listed previously. Then, by a repeated application of the above argument (proof by induction), it follows that for any $t \geq 0$, $v_i[t] = x$ for all $i \in L$, $v_j[t] = X$ for all $j \in R$ and $v_k[t]\in(x,X)$ for all $k\in C$.

Since $L$ and $R$ both contain fault-free nodes, the convergence requirement
is not satisfied. This is a contradiction to the assumption that a correct
iterative algorithm exists.\\
\end{proof}

\section{Sufficient Condition}
\label{sec:sufficiency}

In Theorems \ref{thm:validity} and \ref{thm:convergence} in this section, we prove that Middle Algorithm satisfies {\em validity} and {\em convergence} conditions, respectively, provided that $G(\scriptv,\scripte)$ satisfies the condition below, which matches the necessary condition stated in Theorem \ref{thm:nc}.\\

\noindent{\bf Sufficient condition:}
\begin{itemize}
\item For every node $v \in \scriptv$, $|N_v^-| \geq 3f$, and
\item Let sets $F,L,C,R$ form a partition of $\scriptv$, such that $L$ and $R$ are both non-empty, and $|F|\leq f$. Then, either $C\cup R\Rightarrow L$, or $L\cup C\Rightarrow R$.
\end{itemize}
The claim below follows immediately from the second condition above by
setting $C=\Phi$.
\begin{claim}
\label{claim:nc2}
Suppose that $G(\scriptv,\scripte)$ satisfies the {\em Sufficient} condition stated above. Let $\{F,L,R\}$ be a partition of $\scriptv$, such that  $L$ and $R$ are both non-empty and $|F|\leq f$. Then, either $L\Rightarrow R$ or $R\Rightarrow L$.
\end{claim}

\begin{theorem}
\label{thm:validity}
Suppose that $\scriptf$ is the set of Byzantine faulty nodes, and that $G(\scriptv, \scripte)$ satisfies the {\em sufficient} condition stated above.
Then Middle Algorithm satisfies the {\em validity} condition.
\end{theorem}

\begin{proof}
Consider the $t$-th iteration, and any fault-free node $i\in\scriptv-\scriptf$.
Consider two cases:

\begin{itemize}
\item
$f=0$: In this case, all nodes must be fault-free, and $\scriptf=\Phi$. In (\ref{e_Z}) in Middle Algorithm, note that $v_i[t]$ is computed using states from the previous iteration at node $i$ and other nodes. By definition of $\mu[t-1]$ and $U[t-1]$, $v_j[t-1]\in [\mu[t-1],U[t-1]]$ for all fault-free nodes $j\in\scriptv-\scriptf=\scriptv$. Thus, in this case, all the values used in computing $v_i[t]$ are in the interval $[\mu[t-1],U[t-1]]$. Since $v_i[t]$ is computed as a  weighted average of these values, $v_i[t]$ is also within
$[\mu[t-1],U[t-1]]$.

\item $f>0$: Since $|N_i^-|\geq 3f$, 
$|r_i[t]| \geq 3f$. Thus set $T$ in the {\em Update} step contains at least the largest $f$ values from $r_i[t]$, and set $B$ contains at least the smallest $f$ values from $r_i[t]$. Since at most $f$ nodes are faulty, it follows that, either (i) the values received from the faulty nodes are all eliminated, or (ii) the values from the faulty nodes that still remain are between values received from two fault-free nodes. Thus, the remaining values in $r_i[t]$ -- that is, values received from nodes in set $M$ -- are all in the interval $[\mu[t-1],U[t-1]]$. Also, $v_i[t-1]$ is  in $[\mu[t-1],U[t-1]]$, as per the definition of $\mu[t-1]$ and $U[t-1]$. Thus $v_i[t]$ is computed as a  weighted average of values in $[\mu[t-1],U[t-1]]$, and, therefore, it will also be in $[\mu[t-1],U[t-1]]$.
\end{itemize}

Since $\forall i\in\scriptv-\scriptf$, $v_i[t]\in [\mu[t-1],U[t-1]]$, the validity condition is satisfied.
\end{proof}

\begin{definition}
For disjoint sets $A,B$, $in(A \Rightarrow B)$ denotes the set of all the nodes in $B$ that have at least $1/3$ of the incoming edges from nodes in $A$. More formally,
\[
in(A\Rightarrow B) = \left\{~v~|\,v\in B \mbox{\normalfont~and~}~\frac{|N_v^-\cap A|}{|N_v^-|}~> \frac{1}{3}~\right\}
\]
With an abuse of notation, when $A\not\Rightarrow B$, define $in(A\Rightarrow B)=\Phi$. \\
\end{definition}

\begin{definition}
\label{def:absorb_sequence}
For {\em non-empty disjoint} sets $A$ and $B$, set $A$ is said to {\em propagate to} set $B$ in $l$ steps, where $l>0$, if there exist sequences of sets $A_0,A_1,A_2,\cdots,A_l$ and $B_0,B_1,B_2,\cdots,B_l$ (propagating sequences) such that 

\begin{itemize}
\item $A_0=A$, ~~~~ $B_0=B$, ~~~~ $A_l = A \cup B$, ~~~~ $B_l=\Phi$, 
 ~~~~ $B_\tau \neq \Phi$ ~for~ $\tau<l$, ~~~~~ and
\item for $0\leq \tau\leq l-1$,
\begin{itemize}
\item $A_\tau\Rightarrow B_\tau$, 
\item $A_{\tau+1} = A_\tau\cup in(A_\tau\Rightarrow B_\tau)$, ~~and
\item $B_{\tau+1} = B_\tau - in(A_\tau\Rightarrow B_\tau)$
\end{itemize}
\end{itemize}
\end{definition}
Observe that $A_\tau$ and $B_\tau$ form a partition of $A\cup B$, and for $\tau<l$, $in(A_\tau\Rightarrow B_\tau)\neq \Phi$. Also, when set $A$ propagates to set $B$, the number of steps $l$ in the above definition is upper bounded by $n-1$. \\

\begin{lemma}
\label{lemma:must_absorb}
Assume that $G(\scriptv,\scripte)$ satisfies the {\em sufficient} condition
stated above. For any partition $A,B,F$ of $\scriptv$, where $A,B$ are both non-empty, and $|F|\leq f$, either  $A$ propagates to $B$, or  $B$ propagates to $A$.
\end{lemma}
The proof of Lemma \ref{lemma:must_absorb} is similar to the proof in our prior
work \cite{vaidya_PODC12} -- the proof is included in Appendix \ref{sec:pf_absorb_condition}.

The lemma below states that the interval to which the states at all the
fault-free nodes are confined shrinks after a finite number of iterations of Middle Algorithm. Recall that $U[t]$ and $\mu[t]$ (defined in Section \ref{sec:middle}) are the maximum and minimum over the states at the fault-free nodes at the end of the $t$-th iteration.

\begin{lemma}
\label{lemma:bounded_value}
Suppose that $G(\scriptv, \scripte)$ satisfies the {\em sufficient} condition stated above, and $\scriptf$ is the set of Byzantine faulty nodes. Moreover, at the end of the $s$-th iteration of Middle Algorithm, suppose that the fault-free nodes in $\scriptv-\scriptf$ can be partitioned into non-empty sets $R$ and $L$ such that (i) $R$ propagates to $L$ in $l$ steps, and (ii) the states of nodes
in $R$ are confined to an interval of length $\leq \frac{U[s]-\mu[s]}{2}$. Then, with the {\em Middle} algorithm, 
\begin{eqnarray}
U[s+l]-\mu[s+l]~\leq~\left(1-\frac{\alpha^l}{2}\right)(U[s] - \mu[s])
\label{e:convergence:1}
\end{eqnarray}
where $\alpha$ is as defined in (\ref{e_alpha}).
\end{lemma}
The proof of the above lemma is presented in Appendix \ref{sec:pf_lemma:bounded_value}.

\begin{theorem}
\label{thm:convergence}
Suppose that $\scriptf$ is the set of Byzantine faulty nodes, and that $G(\scriptv, \scripte)$ satisfies the {\em sufficient} condition stated above. Then the {\em Middle} algorithm satisfies the {\em convergence} condition.
\end{theorem}

\begin{proof}
Our goal is to prove that, given any $\epsilon>0$, there exists $\tau$ such that
\begin{equation}
U[t]-\mu[t] \leq \epsilon ~~~\forall t\geq \tau
\end{equation}

Consider $s$-th iteration, for some $s\geq 0$. If $U[s]-\mu[s]=0$, then the algorithm has already converged, and the proof is complete, with $\tau=s$ (recall that we have already proved that the algorithm satisfies the validity condition).

Now, consider the case when $U[s]-\mu[s]>0$. Partition $\scriptv-\scriptf$ into two subsets, $A$ and $B$, such that, for each node $i\in A$,  $v_i[s]\in \left[\mu[s], \frac{U[s]+\mu[s]}{2}\right)$, and for each node $j\in B$,
$v_j[s] \in \left[\frac{U[s]+\mu[s]}{2}, U[s]\right]$. By definition of $\mu[s]$ and $U[s]$, there exist fault-free nodes $i$ and $j$ such that $v_i[s]=\mu[s]$ and $v_j[s]=U[s]$. Thus, sets $A$ and $B$ are both non-empty. By Lemma \ref{lemma:must_absorb}, one of the following two conditions must be true:

\begin{itemize}
\item Set $A$ propagates to set $B$. Then, define $L=B$ and $R=A$. The states of all the nodes in $R=A$ are confined within an interval of length strictly less than $\frac{U[s]+\mu[s]}{2} - \mu[s] \leq \frac{U[s]-\mu[s]}{2}$.

\item Set $B$ propagates to set $A$. Then, define $L=A$ and $R=B$. In this case, states of all the nodes in $R=B$ are confined within an interval of length less than or equal to $U[s]-\frac{U[s]+\mu[s]}{2} \leq \frac{U[s]-\mu[s]}{2}$. 

\end{itemize}

In both cases above, we have found non-empty sets $L$ and $R$ such that (i) $L,R$ is a partition of $\scriptv-\scriptf$, (ii) $R$ propagates to $L$, and (iii) the states in $R$ are confined to an interval of length less than or equal to $\frac{U[s]-\mu[s]}{2}$. Suppose that $R$ propagates to $L$ in $l(s)$ steps, where $l(s)\geq 1$. Then by Lemma~\ref{lemma:bounded_value},

\begin{eqnarray}
U[s+l(s)]-\mu[s+l(s)] \leq \left( 1-\frac{\alpha^{l(s)}}{2}\right)(U[s] - \mu[s])
\label{e_t}
\end{eqnarray}

In the {\em Middle} algorithm, observe that $a_i > 0$ for all $i$. Therefore, $\alpha$ defined in (\ref{e_alpha}) is $> 0$. Then, $n-1 \geq l(s)\geq 1$ and $0<\alpha\leq 1$; hence, $0\leq \left( 1-\frac{\alpha^{l(s)}}{2}\right)<1$.

Let us define the following sequence of iteration indices:
\begin{itemize}
\item $\tau_0 = 0$,
\item for $i>0$, $\tau_i = \tau_{i-1} + l(\tau_{i-1})$, where $l(s)$ for any given $s$ was defined above.
\end{itemize}

If for some $i$,  $U[\tau_i]-\mu[\tau_i]=0$, then since the algorithm is already proved to satisfy the validity condition, we will have $U[t]-\mu[t]=0$ for all $t\geq \tau_i$, and the proof of convergence is complete.

Now, suppose that $U[\tau_i]-\mu[\tau_i]\neq 0$ for the values of $i$ in the
analysis below. By repeated application of the argument leading to (\ref{e_t}), we can prove that, for $i\geq 0$,

\begin{eqnarray}
U[\tau_i]-\mu[\tau_i] \leq \left( \Pi_{j=1}^i\left( 1-\frac{\alpha^{\tau_j-\tau_{j-1}}}{2}\right)\right)~(U[0] - \mu[0])
\end{eqnarray}

For a given $\epsilon$,
by choosing a large enough $i$, we can obtain
\[
\left(\Pi_{j=1}^i\left( 1-\frac{\alpha^{\tau_j-\tau_{j-1}}}{2}\right)\right)~(U[0] - \mu[0]) \leq \epsilon
\]
and, therefore,
\begin{eqnarray}
U[\tau_i]-\mu[\tau_i] \leq  \epsilon
\end{eqnarray}
For $t\geq \tau_i$, by validity of the {\em Middle} algorithm, it follows that
\[
U[t]-\mu[t] \leq
U[\tau_i]-\mu[\tau_i] \leq  \epsilon
\]
This concludes the proof.
\end{proof}

\section{Discussion}
\label{s_discussion}

The results in this report can be easily extended to the following version of the validity condition:
\begin{itemize}
\item {\em Validity}: $\forall t$, $\mu[t]\geq \mu[0]$ and $U[t]\leq U[0]$
\end{itemize}
This validity condition is weaker than the condition satisfied by the {\em Middle} algorithm, therefore, the algorithm satisfies this validity condition as well.
Also, it should be easy to see that our necessary condition also holds under the above validity condition (the proof remains essentially unchanged).

In our analysis here, we assumed that the system is synchronous, and messages sent in each iteration are delivered in the same iteration. That is, the state update in the $t$-th iteration uses neighbors' states at the end of the $(t-1)$-th iteration. The results in this paper can be extended to the case when messages may be delayed such that the latest state available from a neighbor may be from iteration $(t-B)$, for some finite $B>0$. In this case, our original validity condition will need to be modified to require that the state of the fault-free nodes at the end of any iteration remains in the convex hull of the fault-free nodes $B$ iterations ago. 

\comment{====================
Here, we also present a result without showing all the details. We will soon show the complete discussion. We explore middle algorithm in Erd\"{o}s-R\'{e}nyi Random Graphs \cite{erdos_book}. Consider the model, $G_{n, p} = (\scriptv, \scripte)$, where $\scriptv$ contains all the nodes and $|\scriptv| = n$. For every distinct pair of nodes $i, j \in \scriptv$, $(i, j) \in \scripte$ with an independent probability $p$.

Now, the following theorem can be shown.

\begin{theorem}
For every constant positive integer $f$, the threshold function is

\begin{equation}
\label{eq:t}
t=\frac{\ln n + 2f\ln\ln n}{n}
\end{equation}

\end{theorem}

That is, if $p > t$, then the resulting $G_{n, p}$ satisfies the condition presented in this work with high probability, i.e., Middle algorithm solves approximate consensus in the graph; otherwise, $G_{n, p}$ does not satisfy the condition with high probability, i.e., the adversary can find a strategy to break the Middle algorithm.
}

We now state a result without proof. Further details will be presented elsewhere. Consider an Erd\"{o}s-R\'{e}nyi random graphs $G_{n,p}(\sv,\se)$, where $\sv$ contains $n$ vertices, and edge $(i,j)\in\se$ with probability $p$ independently for each $(i,j)$. For large $n$, this random graph satisfies the condition in Theorem \ref{thm:nc} with high probability if and only if $p=\Omega(t)$ where $t$ is a threshold dependent on $n$ and $f$.

\section{Summary}

This paper introduces a parameter-independent iterative algorithm, the {\em Middle} algorithm, that solves the approximate Byzantine consensus problem. The {\em Middle} algorithm does not explicitly use the global parameter of the graph, i.e., the upper-bound on the number of faults, $f$. We prove {\em tight} necessary and sufficient conditions for the correctness of the {\em Middle} algorithm that tolerates up to $f$ Byzantine faults in directed graphs.

\appendix

\section{Proof of Lemma 1}
\label{sec:pf_absorb_condition}

To prove Lemma \ref{lemma:must_absorb}, we first prove the following Lemma.

\begin{lemma}
\label{lemma:absorb_condition}
Assume that $G(\scriptv,\scripte)$ satisfies the {\em Sufficient condition}.
Consider a partition $A,B,F$ of $\scriptv$ such that
$A$ and $B$ are non-empty, and $|F|\leq f$.
If $B \not\Rightarrow A$, then set $A$ propagates to set $B$.
\end{lemma}

\begin{proof}
Since $B\not\Rightarrow A$, by Claim \ref{claim:nc2},
$A\Rightarrow B$.

Define $A_0=A$ and $B_0=B$.
Now, for a suitable $l>0$, we will build propagating sequences $A_0,A_1,\cdots A_l$
and $B_0,B_1,\cdots B_l$ inductively.
\begin{itemize}
\item Recall that $A=A_0$ and $B=B_0\neq \Phi$. Since $A\Rightarrow B$,
$in(A_0\Rightarrow B_0)\neq \Phi$.
Define $A_1=A_0\cup in(A_0\Rightarrow B_0)$
and 
$B_1=B_0-in(A_0\Rightarrow B_0)$.

If $B_1=\Phi$, then $l=1$, and we have found the propagating sequence
already.

If $B_1\neq \Phi$, then define $L=A=A_0$, $R=B_1$ and $C=A_1-A=B-B_1$.
Note that $B=R\cup C$, $A_1=L\cup C$, and $L,C,R,F$ form a partition of the set of nodes.
Since $B\not\Rightarrow A$, $R\cup C\not\Rightarrow L$. Therefore,
by the {\em Sufficient condition}, $L\cup C\Rightarrow R$. That is, $A_1\Rightarrow B_1$.

\item 
For increasing values of $i\geq 0$,
given $A_i$ and $B_i$, where $B_i\neq\Phi$, by following steps similar to the previous
item, we can obtain
$A_{i+1}=A_0\cup in(A_i\Rightarrow B_i)$
and 
$B_{i+1}=B_i-in(A_i\Rightarrow B_i)$,
such that either $B_{i+1}=\Phi$ or $A_{i+1}\Rightarrow B_{i+1}$.
\end{itemize}
In the above construction, $l$ is the smallest index such that
$B_l=\Phi$.
\end{proof}

\paragraph{\bf\large Proof of Lemma \ref{lemma:must_absorb}}

\begin{proof}
Consider two cases:
\begin{itemize}
\item $A\not\Rightarrow B$: Then by Lemma \ref{lemma:absorb_condition} above,
$B$ propagates to $A$, completing the proof.

\item $A\Rightarrow B$: In this case, consider two sub-cases:
\begin{itemize}
\item {\em $A$ propagates to $B$}: The proof in this case is complete.

\item {\em $A$ does not propagate to $B$}:
Recall that $A\Rightarrow B$. Since $A$ does not propagate to $B$,
propagating sequences defined in Definition~\ref{def:absorb_sequence}
do not exist in this case. More precisely, there must exist $k>0$,
and sets $A_0,A_1,\cdots,A_k$ and $B_0,B_1,\cdots,B_k$,
such that:
\begin{itemize}
\item $A_0=A$ and $B_0=B$, and
\item for $0\leq i\leq k-1$,
\begin{list}{}{}
\item[o] $A_i\Rightarrow B_i$,
\item[o] $A_{i+1} = A_i\cup in(A_i\Rightarrow B_i)$, and
\item[o] $B_{i+1} = B_i - in(A_i\Rightarrow B_i)$.
\end{list}
\item $B_{k}\neq \Phi$ \underline{and} $A_{k}\not\Rightarrow B_{k}$.
\end{itemize}
The last condition above violates the requirements for $A$ to propagate
to $B$.

Now, $A_{k}\neq \Phi$, $B_k\neq \Phi$, and $A_k,B_k,F$ form
a partition of $\scriptv$. Since $A_{k}\not\Rightarrow B_{k}$,
by Lemma \ref{lemma:absorb_condition} above,
$B_k$ propagates to $A_k$.

Given that $B_k\subseteq B_0 = B$, $A=A_0\subseteq A_k$, and $B_k$ propagates
to $A_k$, now we prove that $B$ propagates to $A$. 

Recall that $A_i$ and $B_i$ form a partition of $\scriptv-F$.

Let us define $P=P_0=B_k$ and $Q=Q_0=A_k$. Thus, $P$ propagates to $Q$.
Suppose that $P_0,P_1,...P_m$ and $Q_0,Q_1,\cdots,Q_m$ are
the propagating sequences in this case, with $P_i$ and $Q_i$ forming
a partition of $P\cup Q = A_k\cup B_k=\scriptv-F$. \\

Let us define $R=R_0=B$ and $S=S_0=A$.
Note that $R,S$ form a partition of $A\cup B=\scriptv-F$.
Now, $P_0=B_k\subseteq B=R_0$ and $S_0=A\subseteq A_k =Q_0$.
Also, $R_0-P_0$ and $S_0$ form a partition of $Q_0$.

\begin{itemize}
\item
Define $P_1 = P_0\cup (in(P_0\Rightarrow Q_0))$, and $Q_1 = \scriptv - F - P_1 = Q_0 - (in(P_0\Rightarrow Q_0))$. Also,
$R_1 = R_0\cup (in(R_0\Rightarrow S_0))$, and $S_1 = \scriptv - F - R_1 = S_0 - (in(R_0\Rightarrow S_0))$.

Since $R_0-P_0$ and $S_0$ are a partition of $Q_0$,
the nodes in $in(P_0\Rightarrow Q_0)$ belong to one of these
two sets. Note that $R_0-P_0\subseteq R_0$.
Also, $S_0 \cap in(P_0\Rightarrow Q_0) \subseteq in(R_0\Rightarrow S_0)$.
Therefore, it follows that $P_1 = P_0\cup (in(P_0\Rightarrow Q_0))
\subseteq R_0\cup (in(R_0\Rightarrow S_0)) = R_1$.

Thus, we have shown that, $P_1\subseteq R_1$. Then it follows that
$S_1\subseteq Q_1$.
 \\
\item For $0\leq i<m$, let us define $R_{i+1}=R_i\cup in(R_i\Rightarrow 
S_i)$ and $S_{i+1} = S_i - in(R_i\Rightarrow S_i)$. Then following an
argument similar to the above case, we can inductively show that,
$P_i\subseteq R_i$ and $S_i\subseteq Q_i$.
Due to the assumption on the length of the propagating
sequence above, $P_m=P\cup Q=\scriptv-F$ and $Q_m=\Phi$.
Thus, there must exist $r\leq m$, such that for $i<r$, $R_i\neq \scriptv-F$, and $R_r=\scriptv-F$
and $S_r=\Phi$.

The sequences $R_0,R_1,\cdots,R_r$ and $S_0,S_1,\cdots,S_r$ form
propagating sequences, proving that $R=B$ propagates to $S=A$. 
\end{itemize}
\end{itemize}
\end{itemize}
\end{proof}

\section{Proof of Lemma 2}
\label{sec:pf_lemma:bounded_value}

We first present two additional lemmas (using
the notation in Middle Algorithm).

\begin{lemma}
\label{lemma:psi}
Suppose that $\scriptf$ is the set of faulty nodes,
and that $G(\scriptv,\scripte)$ satisfies the ``sufficient condition''
stated in Section \ref{sec:sufficiency}.
Consider node $i\in\scriptv-\scriptf$.
Let $\psi\leq \mu[t-1]$. Then, for $j\in \{i\}\cup M$,
\[
v_i[t] - \psi  ~\geq~  a_i ~ (w_j - \psi)
\]
where $w_j$ is the value received by node $i$ from node $j$ in the $t$-th
iteration.
Specifically, for fault-free $j\in \{i\}\cup M$,
\[
v_i[t] - \psi  ~\geq~  a_i ~ (v_j[t-1] - \psi)
\]
\end{lemma}
\begin{proof}
In (\ref{e_Z}) in Middle Algorithm, for each $j\in\{i\}\cup M$, consider two cases:
\begin{itemize}
\item $j$ is faulty-free: Then, either $j=i$ or  $j\in M\cap (\scriptv-\scriptf)$.
In this case, $w_j=v_j[t-1]$. Therefore,
$\mu[t-1] \leq w_j\leq U[t-1]$.
\item $j$ is faulty: In this case, $f$ must be non-zero (otherwise,
all nodes are fault-free).  
By Theorem \ref{thm:nc}, $|N_i^-|\geq 3f$.
Then it follows that, in step 2 of the {\em Middle} algorithm, $|B|\geq f$, and set $B$ contains the state of at least one fault-free node,
say $k$.
This implies that $v_k[t-1] \leq w_j$.
This, in turn, implies that
$\mu[t-1] \leq w_j.$
\end{itemize}
Thus, for all $j\in \{i\}\cup M$, we have $\mu[t-1] \leq w_j$.
Therefore,
\begin{eqnarray}
w_j-\psi\geq 0 \mbox{\normalfont~for all~} j\in\{i\} \cup M
\label{e_algo_1}
\end{eqnarray}
Since weights in (\ref{e_Z}) in Middle Algorithm add to 1, we can re-write that equation
as,
\begin{eqnarray}
v_i[t] - \psi &=& \sum_{j\in\{i\}\cup M} a_i \, (w_j-\psi) \\
\nonumber
&\geq& a_i\, (w_j-\psi), ~~\forall j\in \{i\}\cup M  ~~~~~\mbox{\normalfont from (\ref{e_algo_1})}
\end{eqnarray}
For fault-free $j\in \{i\}\cup M$, $w_j=v_j[t-1]$, therefore,
\begin{eqnarray}
v_i[t] -\psi &\geq & a_i\, (v_j[t-1]-\psi)
\end{eqnarray}
\end{proof}

~

\begin{lemma}
\label{lemma:Psi}
Suppose that $\scriptf$ is the set of faulty nodes, and that $G(\scriptv,\scripte)$ satisfies the ``sufficient condition''
stated in Section \ref{sec:sufficiency}.
Consider fault-free node $i\in\scriptv-\scriptf$.
Let $\Psi\geq U[t-1]$. Then, for $j\in \{i\}\cup M$,
\[
\Psi - v_i[t] \geq  a_i ~ (\Psi - w_j)
\]
where $w_j$ is the value received by node $i$ from node $j$ in the $t$-th
iteration.
Specifically, for fault-free $j\in \{i\}\cup M$,
\[
\Psi - v_i[t] \geq  a_i ~ (\Psi - v_j[t-1])
\]
\end{lemma}

\begin{proof}
The proof is similar to Lemma \ref{lemma:psi} proof.
\end{proof}

\paragraph{\bf\large Proof of Lemma 2}

\begin{proof}
Since $R$ propagates to $L$, as 
per Definition~\ref{def:absorb_sequence},
there exist sequences of sets
$R_0,R_1,\cdots,R_l$ and $L_0,L_1,\cdots,L_l$, where
\begin{itemize}
\item $R_0=R$, ~~ $L_0=L$, ~~ $R_l=R\cup L$, ~~ $L_l=\Phi$, ~~ for $0\leq \tau<l$, $L_\tau \neq \Phi$, and
\item for $0\leq \tau\leq l-1$,
\begin{list}{}{}
\item[*] $R_\tau\Rightarrow L_\tau$,
\item[*] $R_{\tau+1} = R_\tau\cup in(R_\tau\Rightarrow L_\tau)$, and
\item[*] $L_{\tau+1} = L_\tau - in(R_\tau\Rightarrow L_\tau)$
\end{list}
\end{itemize}
Let us define the following bounds on the states of the nodes
in $R$ at the end of the $s$-th iteration:
\begin{eqnarray}
X & = & max_{j\in R}~ v_j[s] \\ \label{e_M}
x & = & min_{j\in R}~ v_j[s] \label{e_m}
\end{eqnarray}
By the assumption in the statement of Lemma~\ref{lemma:bounded_value},
\begin{eqnarray}
X-x\leq \frac{U[s]-\mu[s]}{2} \label{e_M_m}
\end{eqnarray}
Also, $X\leq U[s]$ and $x\geq \mu[s]$.
Therefore, $U[s]-X\geq 0$ and $x-\mu[s]\geq 0$.

The remaining proof of Lemma~\ref{lemma:bounded_value} relies
on derivation of the three intermediate claims below. \\

\begin{claim}
\label{claim:1}
For $0\leq \tau\leq l$, for each node $i\in R_\tau$,
\begin{eqnarray}
v_i[s+\tau] - \mu[s] ~ \geq~   \alpha^{\tau}(x-\mu[s])
\label{e_ind_1}
\end{eqnarray}
\end{claim}
\noindent{\em Proof of Claim \ref{claim:1}:}
The proof is by induction.

\noindent
{\em Induction basis:}
By definition of $x$, (\ref{e_ind_1}) holds true
for $\tau=0$.

\noindent{\em Induction:}
Assume that (\ref{e_ind_1}) holds true for some $\tau$, $0\leq \tau<l$.
Consider $R_{\tau+1}$.
Observe that $R_\tau$ and $R_{\tau+1}-R_\tau$ form a partition of $R_{\tau+1}$;
let us consider each of these sets separately.
\begin{itemize}
\item Set $R_\tau$: By assumption, for each $i\in R_\tau$, (\ref{e_ind_1})
holds true.
By validity of Middle Algorithm (proved in Theorem \ref{thm:validity}), $\mu[s] \leq \mu[s+\tau]$.
Therefore, setting $\psi=\mu[s]$ and $t=s+\tau+1$ in Lemma~\ref{lemma:psi},
we get,
\begin{eqnarray*}
v_i[s+\tau+1] - \mu[s] & \geq 
& a_i~(v_i[s+\tau] - \mu[s]) \\
& \geq & a_i~ \alpha^{\tau}(x-\mu[s]) ~~~~ \mbox{due to (\ref{e_ind_1})} \\
& \geq & \alpha^{\tau+1}(x-\mu[s])  ~~~~~~ \mbox{due to (\ref{e_alpha})} \\
	&& \mbox{~~ and because~~~~} x-\mu[s]\geq 0 
\end{eqnarray*}

\item Set $R_{\tau+1}-R_\tau$: Consider a node $i\in R_{\tau+1}-R_\tau$. By definition of $R_{\tau+1}$, we have that $i\in in(R_\tau\Rightarrow L_\tau)$.
Thus,

\[ \frac{|N_i^- \cap R_\tau|}{|N_i^-|} > \frac{1}{3} \] 

In Middle Algorithm, values in sets $B$ and $T$ received by
node $i$ are eliminated before $v_i[s+\tau+1]$ is computed at
the end of $(s+\tau+1)$-th iteration. Consider two possibilities:
\begin{itemize}
\item Value received from one of the nodes in $N_i^- \cap R_\tau$ is
{\em not} eliminated. Suppose that this value is received from
fault-free node $p\in N_i^-\cap R_\tau$. Then, $p\in M$, and by an argument similar to the
previous case, we can set $\psi=\mu[s]$
in Lemma~\ref{lemma:psi}, to obtain,
\begin{eqnarray*}
v_i[s+\tau+1] -\mu[s] & \geq & a_i~(v_p[s+\tau]-\mu[s]) \\
& \geq & a_i~ \alpha^{\tau}(x-\mu[s]) ~~~~ \mbox{due to (\ref{e_ind_1})} \\
& \geq & \alpha^{\tau+1}(x-\mu[s])  ~~~~~~ \mbox{due to (\ref{e_alpha})} \\
	&& \mbox{and because~~~~} x-\mu[s]\geq 0 
\end{eqnarray*}

\item
Values received from {\em all} nodes in $N_i^- \cap R_\tau$ are eliminated.
Thus, $(N_i^-\cap R_\tau)\subseteq T\cup B$. Recall that $|N_i^- \cap R_\tau| > |N_i^-|/3\geq |B|=|T|$. Thus, $T$ and $B$
both must contain at least one node from
$N_i^-\cap R_\tau$.
Therefore, the values that are {\em not} eliminated -- that is, values received from nodes in $M$  -- are within the interval to which the values received from the nodes in $N_i^- \cap R_\tau$ belong. Thus, there exists a node $k$ (possibly faulty) in $M$ from whom node $i$ receives
some value $w_k$ -- which is not eliminated -- and
a fault-free node $p\in N_i^- \cap R_\tau$ such that

\begin{eqnarray}
v_p[s+\tau] &\leq & w_k \label{e_wk}
\end{eqnarray}
Then by setting $\psi=\mu[s]$ and $t=s+\tau+1$ in Lemma~\ref{lemma:psi}, we have
\begin{eqnarray*}
v_i[s+\tau+1] -\mu[s] & \geq & a_i~(w_k -\mu[s]) \\
& \geq & a_i~(v_p[s+\tau]-\mu[s])  \mbox{~~~by (\ref{e_wk})} \\
& \geq & a_i~ \alpha^{\tau}(x-\mu[s]) ~~~~ \mbox{due to (\ref{e_ind_1})} \\
& \geq & \alpha^{\tau+1}(x-\mu[s])  ~~~~~~ \mbox{due to (\ref{e_alpha})} \\
	&& \mbox{and because~~~~} x-\mu[s]\geq 0 
\end{eqnarray*}
\end{itemize}
\end{itemize}
Thus, we have shown that for all nodes in $R_{\tau+1}$,
\[
v_i[s+\tau+1] -\mu[s] 
\geq \alpha^{\tau+1}(x-\mu[s]) 
\]
This completes the proof of Claim \ref{claim:1}. \\

\begin{claim}
\label{claim:2}
For each node $i\in \scriptv-\scriptf$,
\begin{eqnarray}
v_i[s+l] - \mu[s] ~ \geq ~  \alpha^{l}(x-\mu[s])
\label{e_ind_2}
\end{eqnarray}
\end{claim}
\noindent{\em Proof of Claim \ref{claim:2}:}
Note that by definition, $R_l = \scriptv-\scriptf$. Then the proof follows by setting $\tau = l$ in the above Claim \ref{claim:1}. \\

\begin{claim}
\label{claim:3}
For each node $i\in \scriptv-\scriptf$,
\begin{eqnarray}
U[s] - v_i[s+l] \geq  \alpha^{l}(U[s]-X)
\label{e_ind_3a}
\end{eqnarray}
\end{claim}

The proof of Claim \ref{claim:3} is similar to the proof of Claim \ref{claim:2}. \\

\noindent
Now let us resume the proof of the Lemma \ref{lemma:bounded_value}.
Thus, 
\begin{eqnarray}
U[s+l] & = & \max_{i\in\scriptv-\scriptf}~ v_i[s+l] \nonumber \\
& \leq & U[s] - \alpha^{l}(U[s]-X) \mbox{~~~~~~~by (\ref{e_ind_3a})}
\label{e_U}
\end{eqnarray}
and
\begin{eqnarray}
\mu[s+l] & = & \min_{i\in\scriptv-\scriptf}~ v_i[s+l] \nonumber \\
& \geq & \mu[s] + \alpha^{l}(x-\mu[s]) \mbox{~~~~~~~by (\ref{e_ind_2}})
\label{e_mu}
\end{eqnarray}
Subtracting (\ref{e_mu}) from (\ref{e_U}),
\begin{eqnarray}
&& U[s+l]-\mu[s+l] \nonumber \\  & \leq & U[s] - \alpha^{l}(U[s]-X)  - \mu[s] - \alpha^{l}(x-\mu[s]) \nonumber \\
&=& (1-\alpha^l)(U[s]-\mu[s]) + \alpha^l(X-x) \nonumber \\
&\leq& (1-\alpha^l)(U[s]-\mu[s]) + \alpha^l~\frac{U[s]-\mu[s]}{2} \nonumber
 \mbox{~~~~by (\ref{e_M_m})} \nonumber \\
&\leq& (1-\frac{\alpha^l}{2})(U[s]-\mu[s])  \nonumber
\end{eqnarray}
This concludes the proof of Lemma~\ref{lemma:bounded_value}.
\end{proof}

\end{document}